\documentclass[11pt,reqnoa,4paper]{amsart}

\usepackage{amscd,amssymb,amsmath,amsthm}
\usepackage{graphicx,epstopdf}
\usepackage{color}
\usepackage{cite}
\usepackage{fullpage}
\usepackage{mathrsfs} %for nicer caligraphic letters via mathscr
\usepackage[dvipsnames]{xcolor}
\usepackage[pagebackref=true,colorlinks=true, linkcolor=RoyalBlue, citecolor=RoyalBlue]{hyperref}

\newtheorem{thm}{Theorem}
\newtheorem{defn}{Definition}
\newtheorem{lemma}{Lemma}
\newtheorem{pro}{Proposition}
\newtheorem{rk}{Remark}

\numberwithin{equation}{section} 

\def\l{\lambda}

\begin{document}
\title[Gibbs measures for hardcore-SOS models on trees]{
Gibbs measures for hardcore-SOS models on Cayley trees}

\author{Benedikt Jahnel \& Utkir Rozikov}

\address{B.\ Jahnel\\ Institut f\"ur Mathematische Stochastik,
	 Technische Universit\"at Braunschweig, Universit\"atsplatz 2, 38106, Braunschweig, Germany \&
	Weierstrass Institute for Applied Analysis and Stochastics, Mohrenstrasse 39, 10117, Berlin, Germany.}
\email {benedikt.jahnel@tu-braunschweig.de}

	\address{U.\ Rozikov$^{1,2,3}$\\ \begin{itemize}
\item[]	$^1$V.I.~Romanovskiy Institute of Mathematics, 9, University str.,  Tashkent, 100174, Uzbekistan.
\item[] $^2$New Uzbekistan University, 1, Movarounnahr str. Tashkent, 100000, Uzbekistan.
\item[] $^3$National University of Uzbekistan,  4, Universitet str.,  Tashkent, 100174, Uzbekistan.
\end{itemize} 
 }
\email{rozikovu@yandex.ru}

\begin{abstract}
We investigate the finite-state $p$-solid-on-solid model, for $p=\infty$, on Cayley trees of order $k\geq 2$ and establish a system of functional equations where each solution corresponds to a (splitting) Gibbs measure of the model. Our main result is that, for three states, $k=2,3$ and increasing coupling strength, the number of translation-invariant Gibbs measures behaves as $1\to3\to5\to6\to7$. This phase diagram is qualitatively similar to the one observed for three-state $p$-SOS models with $p>0$ and, in the case of $k=2$, we demonstrate that, on the level of the functional equations, the transition $p\to\infty$ is continuous.
\end{abstract}
\maketitle



\section{Setting, models and functional equations}
Statistical-mechanics models on trees are known to possess rich structural properties, see for example~\cite{Ro,Ge} for general references. In the present note we contribute to this field by analyzing a hardcore model that emerges as a limit of the $p$-SOS model that was recently introduced in~\cite{CKL}. The setting for this {\em $\infty$-SOS model} is as follows. 

\medskip
Let $\Gamma^k= (V , L)$ be the uniform {\em Cayley tree} where each vertex has $k + 1$ neighbors with $V$ being the set of vertices and $L$ the set of edges. Endpoints $x, y$ of an edge $\ell=\langle x, y\rangle$ are called {\em nearest neighbors}. On the Cayley tree there is a natural distance, to be denoted $d(x,y)$, being the smallest number of nearest-neighbors pairs in a path between the vertices $x$ and $y$, where a {\em path} is a sequence of nearest-neighbor pairs of vertices where two consecutive pairs share at least one vertex. 
For a fixed $x^0\in V$, the {\em root},
we let
$$ V_n=\{x\in V\colon d(x,x^0)\leq n\}\quad \text{ and }\quad W_n=\{x\in V\colon d(x^0,x)=n\}$$
denote the {\em ball} of radius $n$, respectively the {\em sphere} of radius $n$, both with center at $x^0$.
%We will write $x<y$ if the path from $x^0$ to $y$ goes through $x$. This is a partial order on the tree, for example, two different points of $W_n$ can not be ordered by this way. But for any pair of nearest neighbors $x$ and $y$ one has $x<y$ or $x>y$.
Further, let $S(x)$ be the {\em direct successors} of $x$, i.e., for $x\in W_n$
$$S(x)=\{y\in W_{n+1}\colon  d(x,y)=1\}.$$
Next, we denote by $\Phi=\{0,1, \dots, m\}$ the {\em local state space}, i.e., the space of values of the spins associated to each vertex of the tree. Then, a {\em configuration} on the Cayley tree is a collection $\sigma = \{\sigma(x)\colon x\in V\} \in \Phi^V=\Omega$.

\medskip
Let us now describe hardcore interactions between spins of neighboring vertices. For this, let $G=(\Phi,K)$ be a graph with vertex set $\Phi$, the set of spin values, and edge set $K$. A configuration $\sigma$ is called $G$-\textit{admissible} on a
Cayley tree if $\{\sigma (x),\sigma (y)\}\in K$ is an edge of $G$ for any pair of nearest neighbors $\langle x, y \rangle\in L$. We let $\Omega^G$
denote the sets of $G$-admissible
configurations.
The restriction of a configuration on a subset $A$ of $V$ is denoted by $\sigma_A$ and $\Omega_A^G$ denotes the set of all $G$-admissible configurations on $A$.
On a general level, we further define the {\em matrix of activity} on edges of $G$ as a function 
$${\lambda}\colon \{i,j\}\in K
\to \lambda_{i,j}\in \mathbb R_+,$$
where $\mathbb R_+$ denotes the positive real numbers and 
$\l_{i,j}$ is called the {\em activity} of the edge $\{i,j\}\in K$.
In this note, we consider the graph $G$ as shown in Figure~\ref{graph}, which is called a {\em hinge-graph}, see for example~\cite{Br1}.
\begin{figure}[h]
	\begin{center}
\includegraphics[width=9cm]{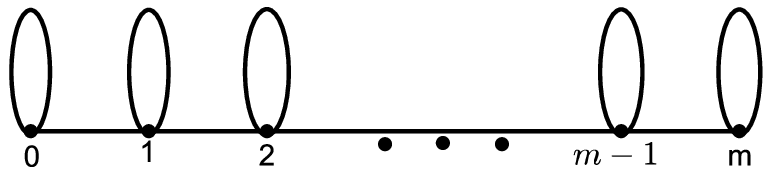} 
	\end{center}
	\caption{The hinge-graph $G$ with $m+1$ vertices.}\label{graph}
\end{figure}
In words, in the hinge-graph $G$, configuration are admissible only if, for any pair of nearest-neighbor vertices $\{ x,y\}$, we have that
\begin{equation}\label{adm}
|\sigma(x)-\sigma(y)|\in \{0,1\}.
\end{equation}
Let us also note that our choice of  admissibilities generalizes certain finite-state random homomorphism models, see \cite{LT,JKR}, where only configurations with $|\sigma(x)-\sigma(y)|=1$ 
are allowed. 

\medskip
Our main interest lies in the analysis of the set of {\em splitting Gibbs measures} (SGMs) defined on hinge-graph addmissible configurations. 
Let us start by defining SGMs for general admissibility graphs $G$. 
Let 
$$z\colon x\mapsto z_x=(z_{0,x}, z_{1,x}, \dots, z_{m,x}) \in \mathbb R^{m+1}_+$$ 
be a
vector-valued function on $V$. Then, given $n=1,2,\ldots$ and an activity $\l=(\lambda_{i,j})_{\{i,j\}\in K}$,
consider the probability distribution  $\mu^{(n)}$ on
$\Omega_{V_n}^G$, defined as
\begin{equation}\label{rus2.1}
\mu^{(n)}(\sigma_n)=\frac{1}{Z_n}\prod_{\langle x, y\rangle\in V_n}\lambda_{\sigma_n(x),\sigma_n(y)} \prod_{x\in
W_n}z_{\sigma(x),x},
\end{equation}
where $\sigma_n=\sigma_{V_n}$. 
Here $Z_n$ is the {\em partition function}
$$
Z_n=\sum_{{\widetilde\sigma}_n\in\Omega^G_{V_n}}
\prod_{\langle x, y\rangle\in V_n}\lambda_{\widetilde\sigma_n(x),\widetilde\sigma_n(y)}\prod_{x\in W_n}
z_{{\widetilde\sigma}(x),x}.
$$
The sequence of probability distributions $(\mu^{(n)})_{n\ge 1}$ is called {\em compatible} if,
for all $n\geq 1$ and $\sigma_{n-1}$, we have that 
\begin{equation}\label{rus2.2}
\sum_{\omega_n\in\Omega^G_{W_n}}
\mu^{(n)}(\sigma_{n-1}\vee\omega_n){\mathbf 1}\{
\sigma_{n-1}\vee\omega_n\in\Omega^G_{V_n}\}=
\mu^{(n-1)}(\sigma_{n-1}),
\end{equation}
where $\sigma_{n-1}\vee\omega_n$ is the concatenation of the configurations $\sigma_{n-1}$ and $\omega_n$. 
Note that, by Kolmogorov's extension theorem, for a compatible sequence of distributions, there exists a unique measure $\mu$ on $\Omega^G$ such that,
for all $n$ and $\sigma_n\in \Omega^G_{V_n}$,
$$\mu(\{\sigma|_{V_n}=\sigma_n\})=\mu^{(n)}(\sigma_n).$$
This motivates the following definition. 
\begin{defn} We call the measure $\mu$, defined by \eqref{rus2.1} and \eqref{rus2.2}, the {\em splitting Gibbs measure} corresponding to the activity $\lambda$ and the function $z\colon x\in V
\setminus\{x^0\}\mapsto z_x$.
\end{defn}
%Let us present our first result which exhibits a condition for compatibility. 
Let $A=
A^G=\big(a_{ij}\big)_{\{i,j\}\in K}$ denote the adjacency matrix of
$G$, i.e.,
$$ a_{ij}= a^G_{ij}=\left\{\begin{array}{ll}
1,\ \ \mbox{if}\ \ \{i,j\}\in K,\\
0, \ \ \mbox{if} \ \  \{i,j\}\notin K,
\end{array}\right.$$
then, the following statement describes conditions on $z_x$ guaranteeing
compatibility of the distributions $(\mu^{(n)})_{n\ge 1}$.

\begin{thm}\label{rust1} The sequence of probability
distributions $(\mu^{(n)})_{n\ge 1}$ in \eqref{rus2.1} are 
compatible if and only if, for any $x\in V$, the following system of equations
holds
\begin{equation}\label{jr}
z_{i,x}= \prod_{y\in S(x)}{\sum_{j=0}^{m-1}a_{ij}\lambda_{i,j}z_{j,y}+a_{im}\lambda_{i,m}\over
\sum_{j=0}^{m-1}a_{mj}\lambda_{m,j}z_{j,y}+a_{mm}\lambda_{m,m}}, \ \ i=0,1,\dots,m-1.
\end{equation}
%where $z'_{i,x}=z_{i,x}/z_{m,x}$.\
\end{thm}
\begin{proof}
The proof is similar to the proof of~\cite[Theorem 1]{Ros} and~\cite[Proposition 2.1]{Ro12}.  
\end{proof}
Note that, in~\eqref{jr}, the normalization is at the spin state $m$, i.e., we assume that, for all $x\in V$, we have  $z_{m,x}=1$.

\medskip
In the remainder of the manuscript, we restrict our choice of activities in order to make contact to $p$-SOS models defined via the formal Hamiltonian
\begin{align}\label{def_pSOS}
 H(\sigma)=-J\sum_{\langle x,y\rangle}
|\sigma(x)-\sigma(y)|^p,
\end{align}
for $p>0$ and coupling constant $J\in \mathbb{R}$, see
\cite{CKL,Ge,Ro,RoP} and references therein. The present note then presents a continuation of previous investigations related to $p$-SOS models on trees with $p>0$, but now in the case where $p=\infty$.
%We note that such admissible configurations are suitable to consider the formal Hamiltonian of $p$-SOS model, when $p=\infty$:
More precisely, we denote $\theta=\exp(J)$ and consider the activity $\l=(\lambda_{i,j})_{\{i,j\}\in K}$ defined as
\begin{equation}\label{isos}
\lambda_{i,j}=\left\{\begin{array}{lll}
1, \ \ \mbox{if} \ \ i=j\\[2mm]
\theta, \ \ \mbox{if} \ \ |i-j|=1,\\[2mm]
\infty, \ \ \mbox{otherwise.}
\end{array}
\right.
\end{equation}
We call the resulting hinge-graph model (with hinge-graph as in  Figure~\ref{graph}) the {\em $\infty$-SOS model}. Consequently, for the $\infty$-SOS model with  activity~\eqref{isos}, the equation~\eqref{jr} reduces to
\begin{equation}\label{equa}
\begin{array}{lll}
z_{0,x}=\prod_{y\in S(x)}{z_{0,y}+\theta z_{1,y}\over \theta z_{m-1,y}+1}\\[3mm]
z_{i,x}=\prod_{y\in S(x)}{\theta z_{i-1,y}+z_{i,y}+\theta z_{i+1,y}
    \over \theta z_{m-1,y}+1}, \ \ i=1,\dots, m-1\\[3mm]
  z_{m,x}= 1,
    \end{array}
   \end{equation}
and, by Theorem~\ref{rust1}, for any $z=\{z_x\colon x\in V\}$
satisfying~\eqref{equa}, there exists a unique SGM $\mu$ for the $\infty$-SOS. However,
the analysis of solutions to~\eqref{equa} for an arbitrary $m$ is challenging. We therefore restrict our attention to a smaller class of measures, namely the translation-invariant SGMs.

\section{Translation-invariant SGMs for the $\infty$-SOS model with $m=2$}
Searching only for translation-invariant measures, the functional equation~\eqref{equa} reduces to 
\begin{equation}
\label{sti}
\begin{array}{lll}
z_{0}=\left({z_{0}+\theta z_{1}\over \theta z_{m-1}+1}\right)^k,\\[3mm]
z_{i}=\left({\theta z_{i-1}+z_{i}+\theta z_{i+1}
    \over \theta z_{m-1}+1}\right)^k, \ \ i=1,\dots, m-1,\\[3mm]
  z_{m}= 1.
    \end{array}
 \end{equation}
In the following we restrict our attention to the case where $m=2$. In this case, denoting $x=\sqrt[k]{z_0}$ and $y=\sqrt[k]{z_1}$, from ~\eqref{sti} we get 
\begin{equation}\label{si}
x={x^k+\theta y^k\over \theta y^k+1}\qquad \text{ and }\qquad
y={\theta x^k+y^k+\theta \over \theta y^k+1}.
     \end{equation}
In particular, considering only the first equation of this system, we find the solutions $x=1$ and 
\begin{equation}\label{ky} 
\theta y^k=x^{k-1}+x^{k-2}+\dots +x.
\end{equation}
We start by investigating the case $x=1$. 
\subsection{Case $x=1$}
In this case, from the second equation in \eqref{si}, we get that
\begin{equation}\label{yk}
    \theta y^{k+1}-y^k+y-2\theta=0
\end{equation}
and hence, as a direct application of Descartes' rule of signs, the following statement follows.
\begin{lemma} For all $k\geq 2$, there exist at most three positive roots for~\eqref{yk}.
\end{lemma}
 
For small values of $k$, i.e., $k=2,3$, we can solve~\eqref{yk} explicitly and exhibit regions of $\theta$ where there are exactly three solutions. 

\subsubsection{Case $k=2$} 
In this case, the equation~\eqref{yk} coincides with the corresponding equation for $p=\infty$ found in~\cite{JR}. The solutions are presented in Section~\ref{sec_4} below. 
\subsubsection{Case $k=3$}\label{sec_13} In this case, we can rearrange~\eqref{yk} as
$$\theta={y^3-y\over y^4-2}=:\alpha(y),$$
where we assumed $y\ne \sqrt[4]{2}$ since $y=\sqrt[4]{2}$ is not a solution. 
It follows from Figure~\ref{theta} that up to three solutions $y$ for~\eqref{yk} appear as follows. 
\begin{figure}[h]
\begin{center}
\includegraphics[width=6.5cm]{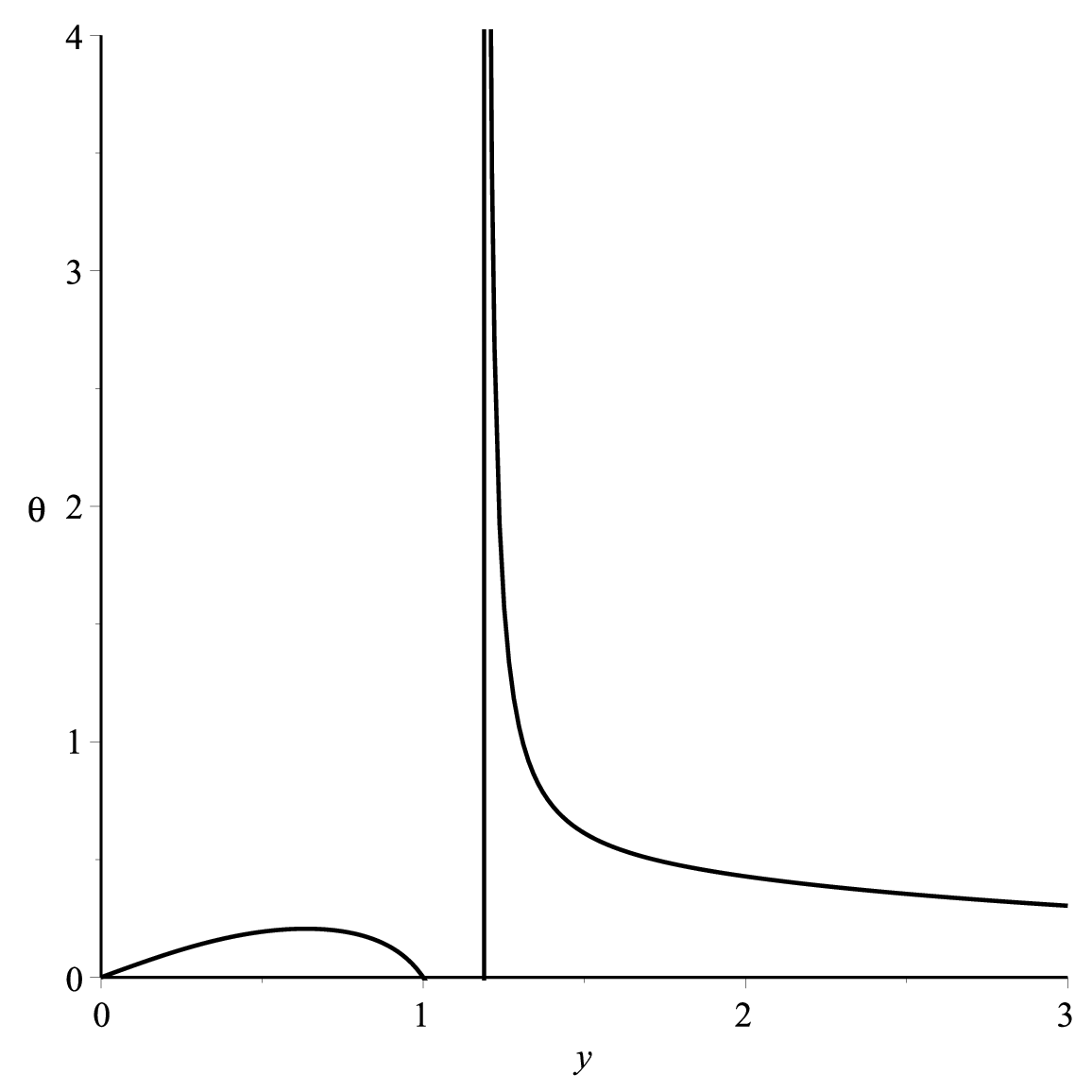} \ \ \ \ \
	\end{center}
	\caption{Graph of the function $\alpha(y)$.}\label{theta}
\end{figure}
There exists $\theta_c\approx 0.206$ such that we have:
\begin{itemize}
\item[1)] If $\theta\in (0,\theta_c)$, then there are three solutions $y_1>\sqrt[4]{2}$ and $y_2<y_3<1$.
\item[2)] If $\theta=\theta_c$, then there are two solutions $y_1>\sqrt[4]{2}$ and $y_2<1$.
\item[3)] If $\theta>\theta_c$, then there is a unique solution $y_1>\sqrt[4]{2}$.
\end{itemize}
We can derive the value of $\theta_c$ explicitly as follows. The equation $\alpha'(y)=0$ can be solved explicitly as  
$$y_0=\sqrt{1-c+1/c}\approx 0.635, \ \ \mbox{with} \ \ c=\sqrt[3]{1+\sqrt{2}}.$$
Since $y_0\in (0, 1)$, the critical value $\theta_c$ is thus given by
\begin{equation}\label{tc}\theta_c=\alpha(y_0)={(1-c)\sqrt{c(1+c-c^2)}\over c^4-3c^2+2c-2\sqrt{2}-1}.
\end{equation}

\subsection{Case $x\ne 1$}
Let us consider the second situation as presented in~\eqref{si}.
\subsubsection{Case $k=2$}
In this case, using ~\eqref{ky} and the second equation in~\eqref{si}, we get
\begin{equation}\label{2x}
\theta^4x^4+(2\theta^2-\theta)x^3+(2\theta^4-2\theta+1)x^2+(2\theta^2-\theta)x+\theta^4=0,
\end{equation}
which is a polynomial with symmetric coefficients and hence, denoting $\xi=x+1/x$,~\eqref{2x} can be rewritten as
$$ \theta^4(\xi^2-2)+(2\theta^2-\theta)\xi+(2\theta^4-2\theta+1)=0, $$
which is equivalent to
\begin{equation}\label{2xi}
\theta^4 \xi^2+(2\theta^2-\theta)\xi+1-2\theta=0.
\end{equation}
But, this equation has two solutions $\xi_{1,2}$ given in~\eqref{12i} below 
%and conditions on parameters are the same as in the next section. 
and thus we have up to four additional solutions. Again, this case coincides with the corresponding equation for $p=\infty$ found in~\cite{JR} and the overview of solutions is presented in Section~\ref{sec_4} below.
% In other words, we have demonstrated that using equations ~\eqref{equa} directly yields the same results as fixing $p$ and then letting $p$ approach infinity.\\

\subsubsection{Case $k=3$} In this case, by~\eqref{ky} and the second equation in~\eqref{si} we obtain 
\begin{equation}\label{x8}
\begin{split}
\theta^6 x^8+(-\theta^6+3\theta^4-\theta^2)x^7+&\theta^6x^6+(2\theta^6-3\theta^2+1)x^5+(-2\theta^6+6\theta^4-7\theta^2+2)x^4\\
&+(2\theta^6-3\theta^2+1)x^3+\theta^6x^2+(-\theta^6+3\theta^4-\theta^2)x+\theta^6=0
\end{split}
\end{equation}
and this equation may have up to eight positive solutions since, if $\theta<\theta_c''\approx 0.605$, the number of sign changes in the coefficients is eight. However, by computer analysis we can show that there exists $\widehat \theta_c\approx 0.4812$ such that, if $\theta<\widehat\theta_c$, then~\eqref{x8} has precisely four positive solutions, if $\theta=\widehat\theta_c$, then there are precisely two solutions and, if $\theta>\widehat\theta_c$, then there exists no positive solution (see Figure~\ref{x8c} for an implicit plot). For example, if $\theta=0.481$, then the four positive solutions are given approximately by 
$$ 0.2072006567,0.2260627940,4.423549680, \text{ and  }4.826239530.$$
\begin{figure}[h]
	\begin{center}
\includegraphics[width=7.3cm]{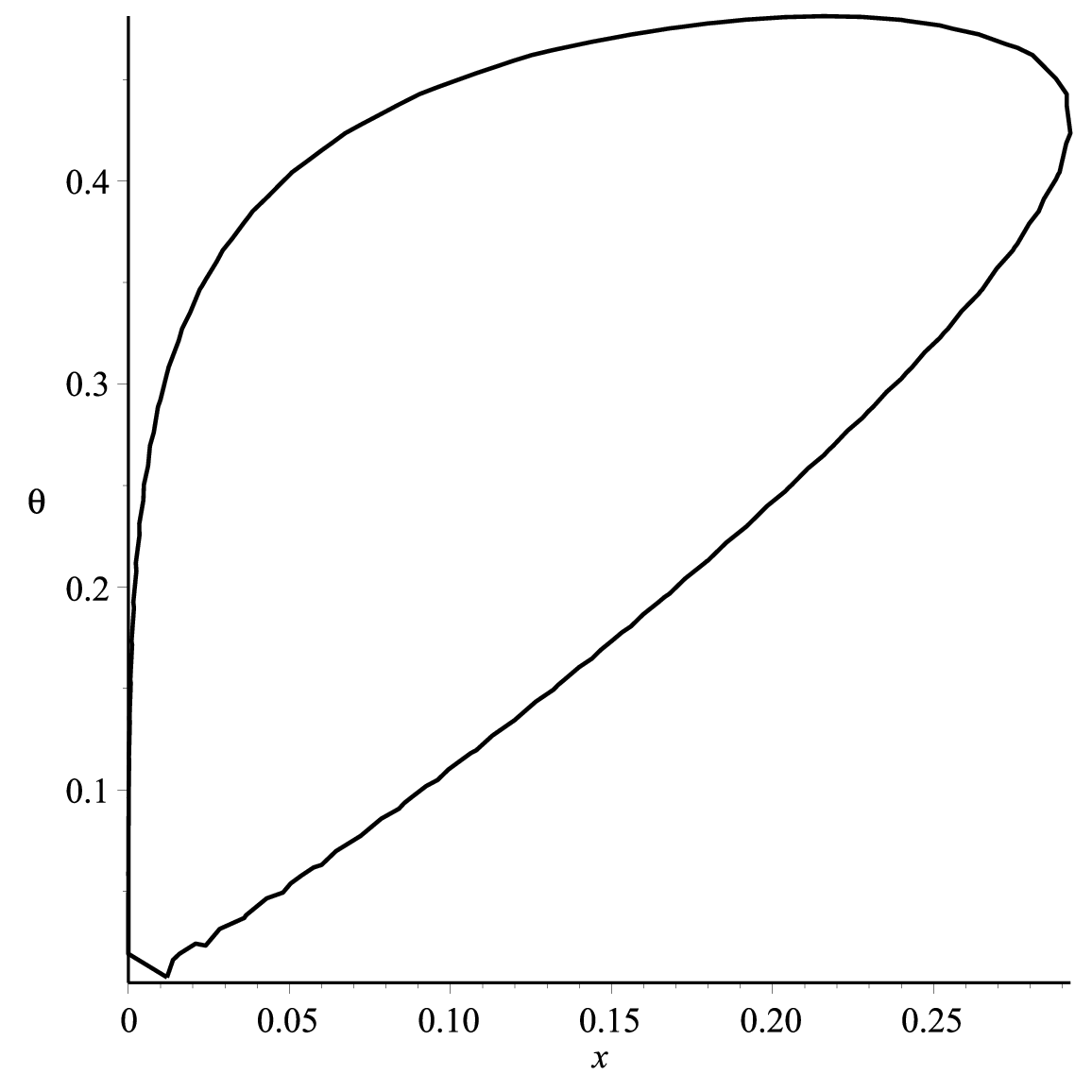}
\includegraphics[width=7.3cm]{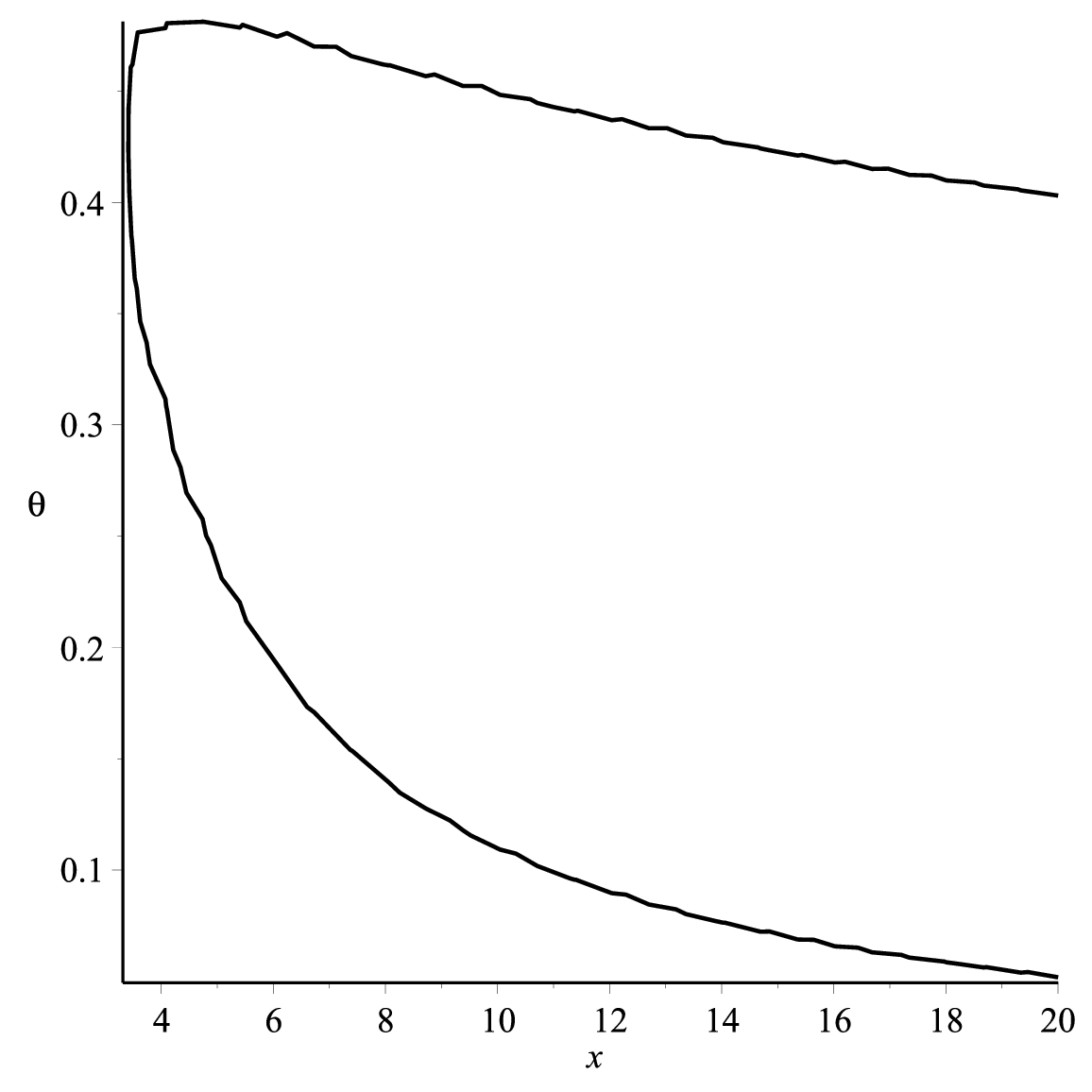}
	\end{center}
\caption{The implicit plot of~\eqref{x8}. Left: For $\theta\in (0, \widehat\theta_c)$, there are two solutions in $(0, 0.2927)$. Right: For the same region of $\theta$ there are two fixed points in $(3.41, \infty)$ since both curves towards the right side converge to zero.}\label{x8c}
\end{figure}

Let us provide an analytical proof of this as well.
Since~\eqref{x8} is symmetric, we divide both sides by $x^4$ and denote $\eta=x+1/x$. Then, we arrive at 
$$\theta^6((\eta^2-2)^2-2)+(-\theta^6+3\theta^4-\theta^2)\eta(\eta^2-3)+\theta^6(\eta^2-2)+(2\theta^6-3\theta^2+1)\eta-2\theta^6+6\theta^4-7\theta^2+2=0$$
and consequently
\begin{equation}\label{eta}
\theta^6\eta^4+(-\theta^6+3\theta^4-\theta^2)\eta^3-3\theta^6\eta^2+(5\theta^6-9\theta^4+1)\eta-2\theta^6+6\theta^4-7\theta^2+2=0.
 \end{equation}
Thus, each solution $\eta>2$ of~\eqref{eta} defines two positive solutions to~\eqref{x8}.  
Denoting $t=\theta^2$, we write~\eqref{eta} as the following cubic equation with respect to $t$, 
 $$
 (\eta^4-\eta^3-3\eta^2+5\eta-2)t^3+(3\eta^3-9\eta+6)t^2+(-\eta^3-7)t+\eta+2=0,
 $$
which is equivalent to
 \begin{equation}\label{t}
 (\eta-1)^3(\eta+2)t^3+3(\eta-1)^2(\eta+2)t^2-(\eta^3+7)t+\eta+2=0.
 \end{equation}
 Now, again by the rule of sign changes, for each $\eta>2$, this equation may have up to two positive solution $t=\theta^2=t(\eta)$. 
For $\eta>2$ let us introduce the new variables
\begin{equation}\label{wE}
w=(\eta-1)t>0\qquad \text{ and }\qquad E={\eta^3+7\over (\eta-1)(\eta+2)}=:b(\eta).
\end{equation}
With this,~\eqref{t} can be expressed as
\begin{equation}
    w^3+3w^2-Ew+1=0
\end{equation}
and we can solve the last equation with respect to $E$, which leads to
$$E=w^2+3w+1/w=:a(w).$$
Note that the function $a(w)$ is monotone decreasing between $0$ and $1/2$ (because $a'(w)=0$ has a unique positive solution $w=1/2$ and $a(0)=a(+\infty)=\infty$) and increasing when $w>1/2$. Thus, the minimal value of $E$ is given by $a(1/2)=15/4$ and hence, for each $E>15/4$, there are exactly two positive solutions $w_1=w_1(E)$ and $w_2=w_2(E)$, with $w_1<w_2$. If $E=15/4$, then there exists a unique $w_c=1/2$. If $E<15/4$, then there is no solution, see Figure~\ref{E}. 
\begin{figure}[h]
	\begin{center}
\includegraphics[width=7.5cm]{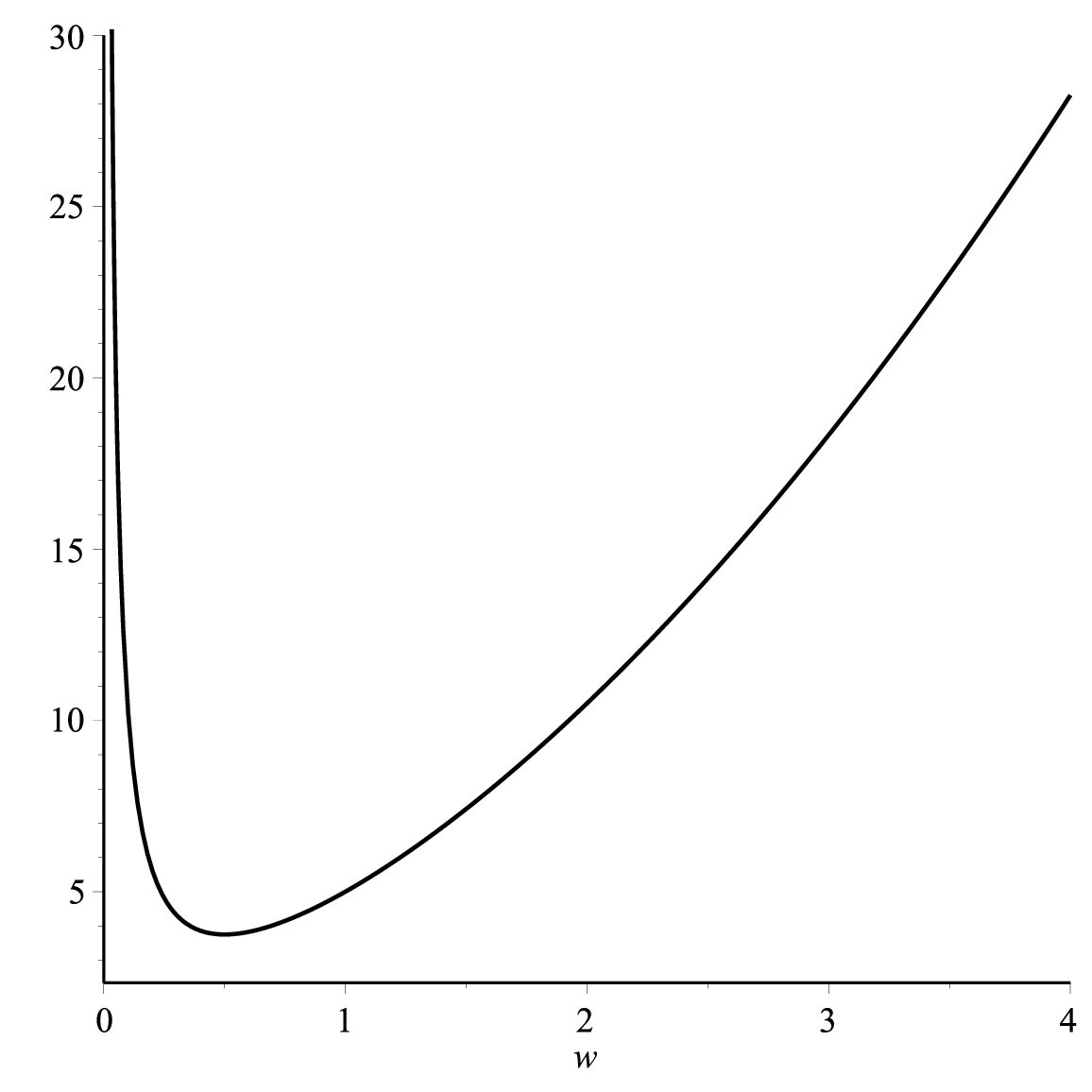} 
	\end{center}
	\caption{The graph of the function $E=a(w)$.}\label{E}
\end{figure}
From Figure~\ref{E} it is also easy to see that $w_1(E)$ is a decreasing function with values in $(0, 1/2]$ and $w_2(E)$ is an increasing function with values in $[1/2, \infty)$ with respect to the variable $E>15/4$. Moreover, 
$$\lim_{E\to\infty}w_1(E)=0\qquad \text{ and }\qquad\lim_{E\to\infty}w_2(E)=\infty,$$
and from the function $b(\eta)$ given in~\eqref{wE}, for $E=15/4$, we get that $$\eta_c=(7+3\sqrt{57})/8\approx 3.707.$$ 
Note that, for the derivative of $\eta\mapsto b(\eta)$, we have 
$$b'(\eta)={(\eta+1)^2(\eta^2-7)\over (\eta-1)^2(\eta+2)^2}.$$
Consequently, $b(\eta)$ is increasing for $\eta>\sqrt{7}\approx 2.65$
with a minimum value $b(\sqrt{7})=21/(1+2\sqrt{7})\approx 3.34$. For each $E\geq 15/4$, since $b(\eta)$ is an increasing function,  from $E=b(\eta)$ one obtains a unique $\eta\geq \eta_c$.  
As a consequence, $w_1(b(\eta))$ is a decreasing and $w_2(b(\eta))$ is an increasing functions of $\eta>\eta_c$.
Thus, $\eta$ is a solution to the following two independent equations, obtained from the first formula of~\eqref{wE},
\begin{equation}\label{w1}
w_1(b(\eta))=(\eta-1)\theta^2.
\end{equation}
\begin{equation}\label{w2}
w_2(b(\eta))=(\eta-1)\theta^2.
\end{equation}
\begin{figure}[h]
	\begin{center}
\includegraphics[width=7.3cm]{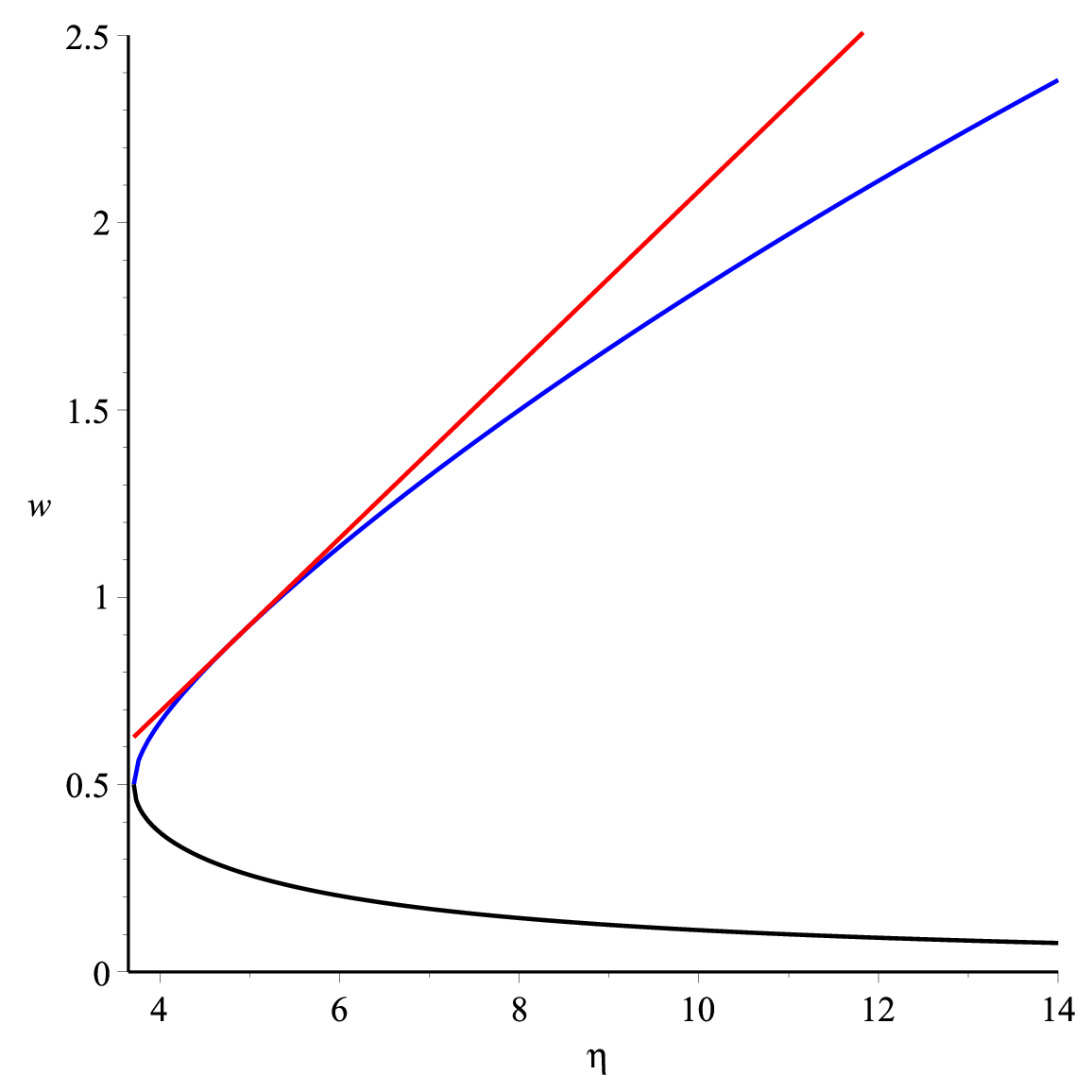}
\includegraphics[width=7.3cm]{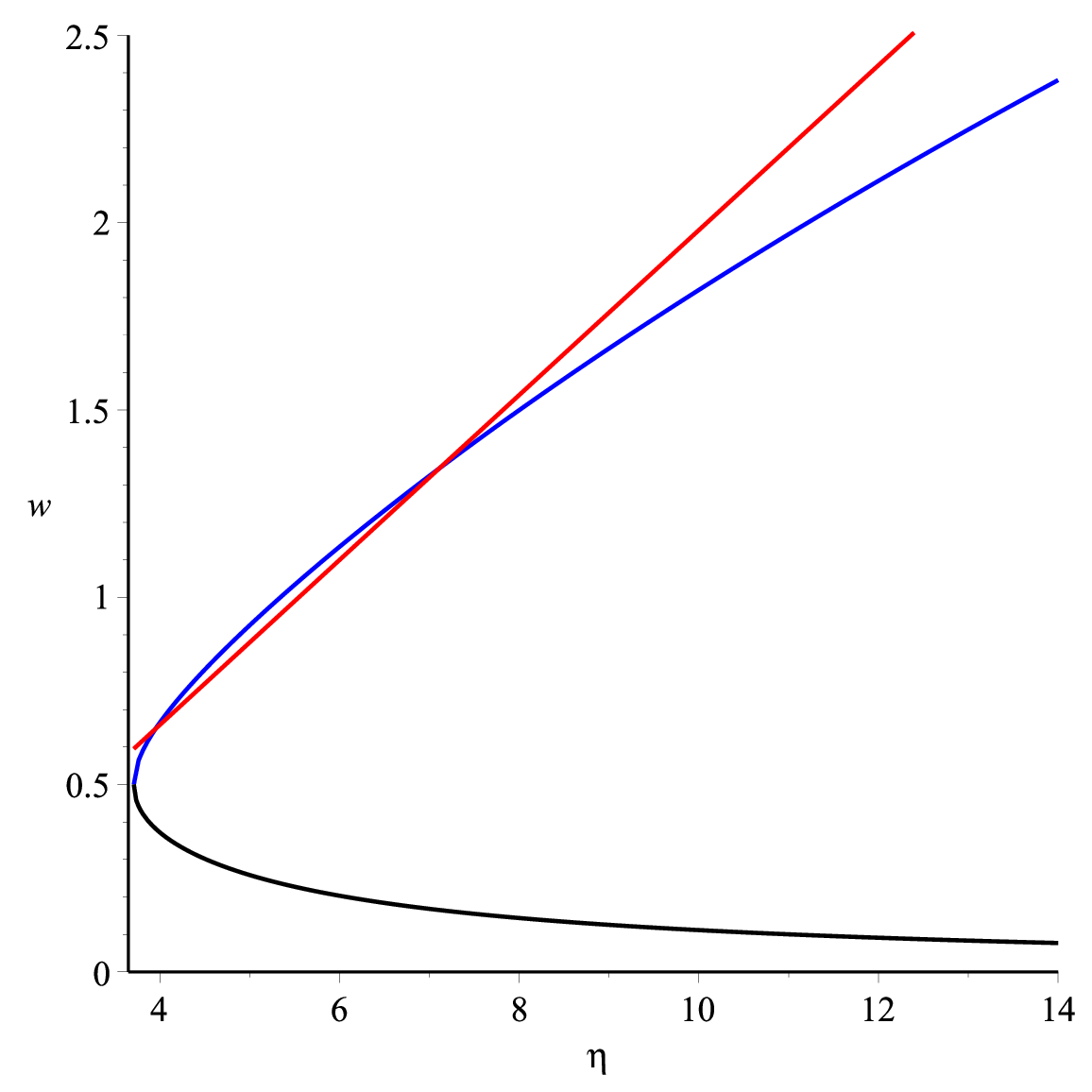}
\includegraphics[width=7.3cm]{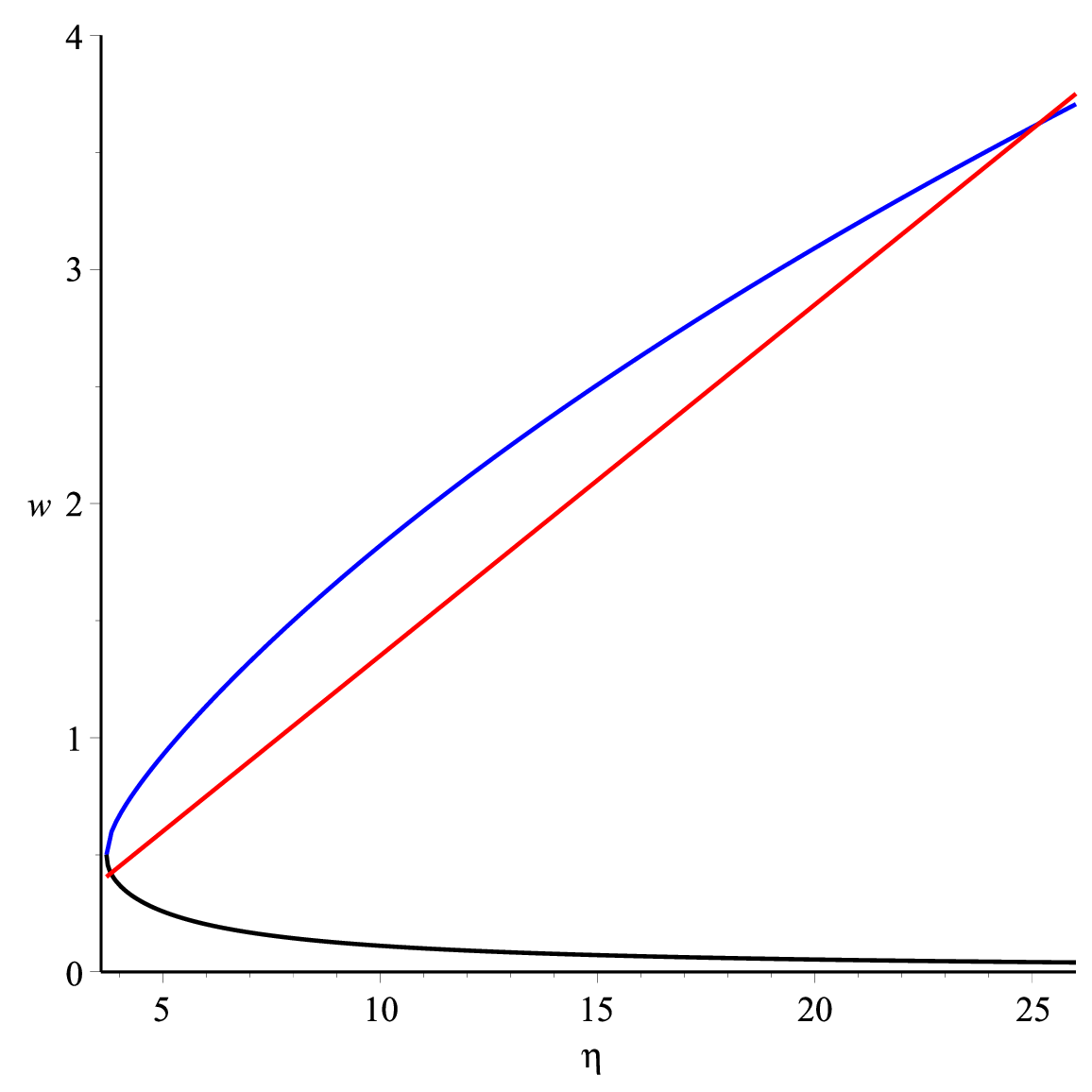}
	\end{center}
\caption{Up-left: The graphs of $\eta\mapsto(\eta-1)\widehat\theta_c^2$ (red),$\eta\mapsto w_2(b(\eta))$ (blue) and  $\eta\mapsto w_1(b(\eta))$ (black). Up-right: The graphs of $\eta\mapsto(\eta-1)0.22$ (red), $\eta\mapsto w_2(b(\eta))$ (blue) and  $\eta\mapsto w_1(b(\eta))$ (black). Down:  The graphs of $\eta\mapsto(\eta-1)0.15$ (red), $\eta\mapsto w_2(b(\eta))$ (blue) and $w_1(b(\eta))$ (black).}\label{w12c}
\end{figure}
Denoting by $\tilde\theta$ the positive solution of the equation
$$1/2=(\eta_c-1)\tilde\theta^2,$$
see \eqref{wE}, we have
$$\tilde\theta={2\over \sqrt{3\sqrt{57}-1}}\approx 0.4298.$$
Then, as can be seen in Figure~\ref{w12c}, the following assertions hold.
\begin{itemize}
\item[i.] If $\theta\in (0, \tilde\theta)$, then~\eqref{w1} and~\eqref{w2} have a unique solution $\eta>2$.
\item[ii.] If $\theta\in [\tilde\theta, \widehat\theta_c)$, then~\eqref{w1} has no solution, but~\eqref{w2} has two solutions greater than 2.
\item[iii.] If $\theta=\widehat\theta_c$, then~\eqref{w1} has no solution, but~\eqref{w2} has a unique solution.
\item[iv.] If $\theta>\widehat\theta_c$, then both equations have no solution.
\end{itemize}
Consequently, under the above mentioned Conditions 
i.-ii.,~\eqref{eta} has two solutions greater then 2, which define four positive solutions for~\eqref{x8}.

\begin{rk}To find the exact critical value $\widehat\theta_c$ mentioned above, one has to solve the following system of equation with respect to unknowns $\widehat\theta_c$ and $\eta_1$
$$w_2'(b(\eta_1))b'(\eta_1)=\widehat\theta_c^2\qquad \text{ and }\qquad w_2(b(\eta_1))=(\eta_1-1)\widehat\theta_c^2.$$
\end{rk}
With respect to Theorem~\ref{rust1}, we can summarize our results for $k=3$ in the following statement.
\begin{thm}\label{tii}
	For the $\infty$-SOS model with $m=2$ and $k=3$, there exist critical values $\theta_c\approx 0.206$ (given explicitly by ~\eqref{tc}) and $\widehat\theta_c\approx 0.4812$ such that 
	\begin{itemize}
		\item[1.] If $\theta>\widehat\theta_c$, then there is unique translation-invariant SGM.
		\item[2.] If  $\theta=\widehat\theta_c$, then there are three  translation-invariant SGMs.
		\item[3.] If  $\theta\in (\theta_c, \widehat\theta_c)$, then there are five  translation-invariant SGMs.
		\item[4.] If  $\theta=\theta_c$, then there are six  translation-invariant SGMs.
		\item[5.] If  $\theta\in (0, \theta_c)$, then there are seven translation-invariant SGMs.
	\end{itemize} 
\end{thm}

\section{The $p$-SOS model with $p\to \infty$}\label{sec_4}

In this section we exhibit the functional equations corresponding to the $p$-SOS model defined by the Hamiltonian~\eqref{def_pSOS} and give results related to the case when $p\to\infty$. The limiting equations turn out to be the same equations as the ones for the $\infty$-SOS model.
Assuming $k=m=2$,~\cite{JR} establishes and analyzes the translation-invariant SGMs of the $p$-SOS model corresponding to the positive solutions of the following system
\begin{equation}\label{rs3.2a}
x={x^2+\theta y^2+\theta^{2^p} \over \theta^{2^p}x^2+\theta y^2+1},
\end{equation}
\begin{equation}\label{rs3.2b}
 y={\theta x^2+y^2+\theta \over \theta^{2^p}x^2+\theta y^2+1}.
\end{equation}
%Further, in~\cite{JR}, a  detailed analysis of this system for fixed $p>0$ is presented. 
In the following, we establish limits for the obtained solutions when $p\to \infty$.
First, from~\eqref{rs3.2a} we get $x=1$ or
\begin{equation}\label{y}
\theta y^2=(1-\theta^{2^p})x-\theta^{2^p}(x^2+1).
\end{equation}
\begin{rk}\label{<1} 
Since $x>0$ we have that~\eqref{y} can hold iff $\theta<1$.
\end{rk}

\subsection{Case $0<\theta<1$.} Let us distinguish two subcases.
\subsubsection{Case $x=1$} Solving by computer the cubic equation
\begin{equation}\label{y3}
\theta y^3-y^2+(\theta^{2^p}+1)y-2\theta=0
\end{equation}
and taking limits of each solution as $p\to\infty$, we see that the solutions have the limits $y_i(\theta)$, $i=1,2,3$. The limiting functions have lengthy formulas, but their graphs can be simply plotted as shown in Figure~\ref{y1-y3-p}.
\begin{figure}[h]
	\begin{center}
\includegraphics[width=7.5cm]{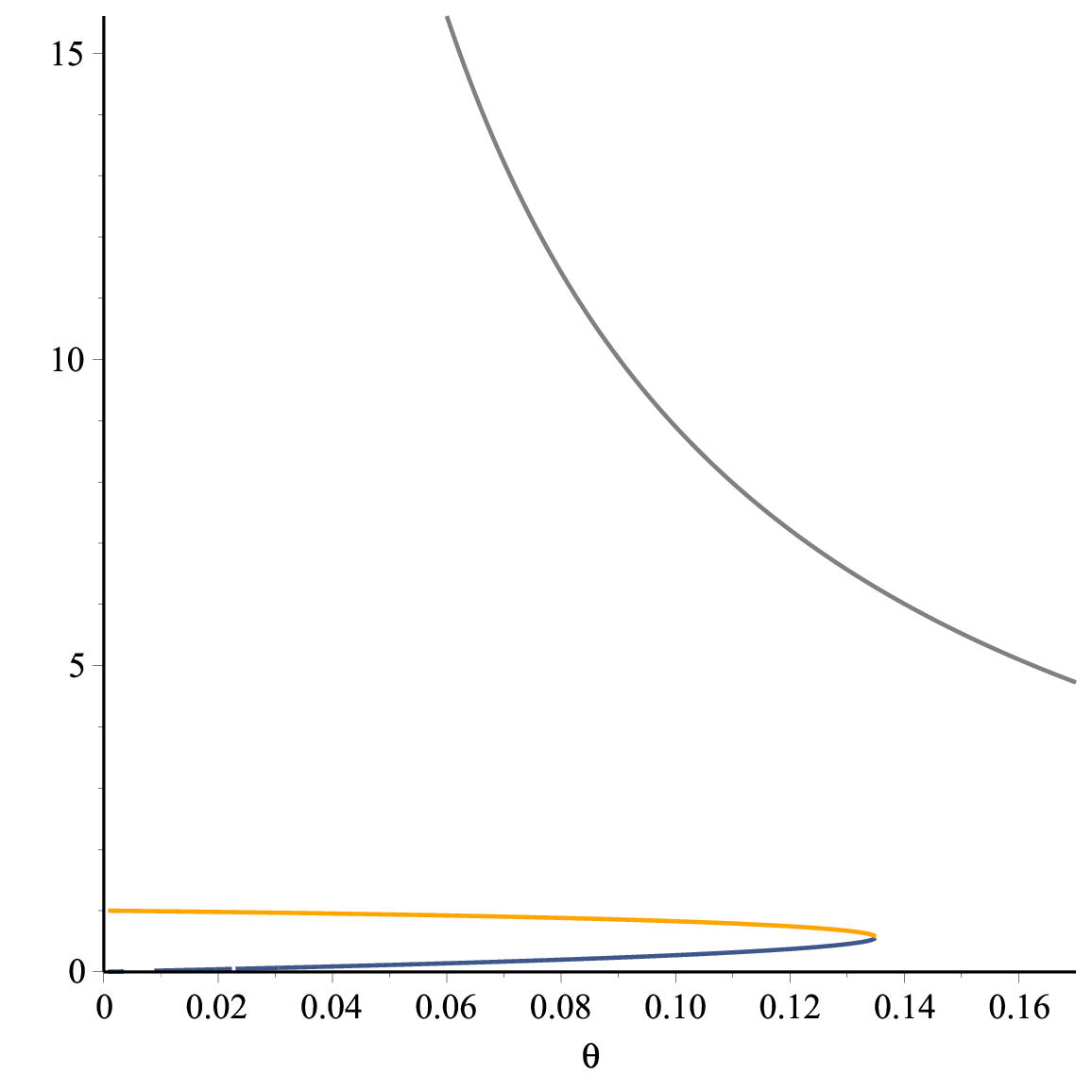} \ \ \ \ \
	\includegraphics[width=7.5cm]{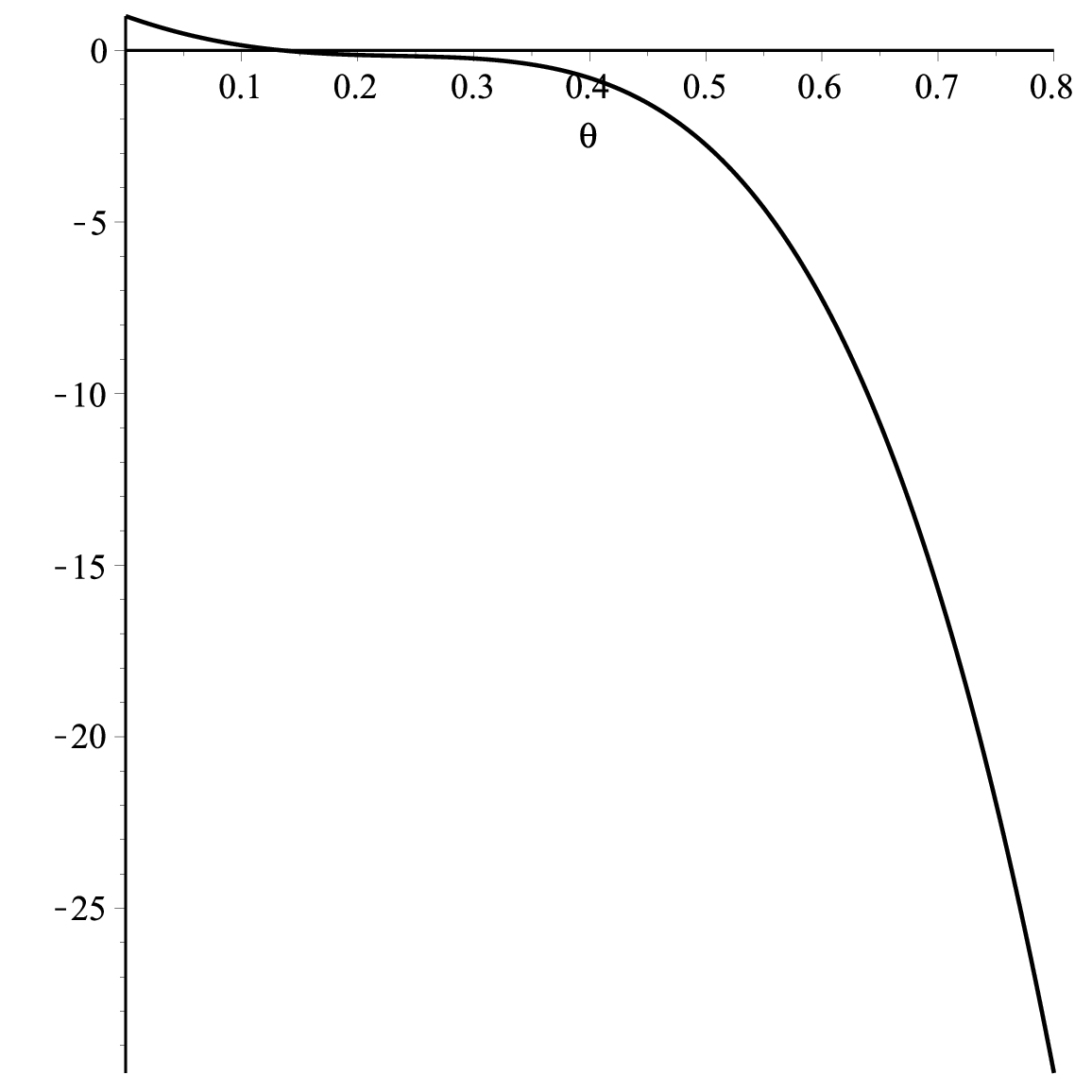} \ \ \
	\end{center}
	\caption{Left: The graphs of the functions $\theta\mapsto y_1(\theta)$ (gray), $\theta\mapsto y_2(\theta)$ (blue) and  $\theta\mapsto y_3(\theta)$ (orange). Right: The graph of the function $\theta\mapsto \Delta_0(\theta)$ for $\theta\in(0,1)$.}\label{y1-y3-p}
%\end{figure}
%\begin{figure}[h]
	\begin{center}
		\includegraphics[width=6.5cm]{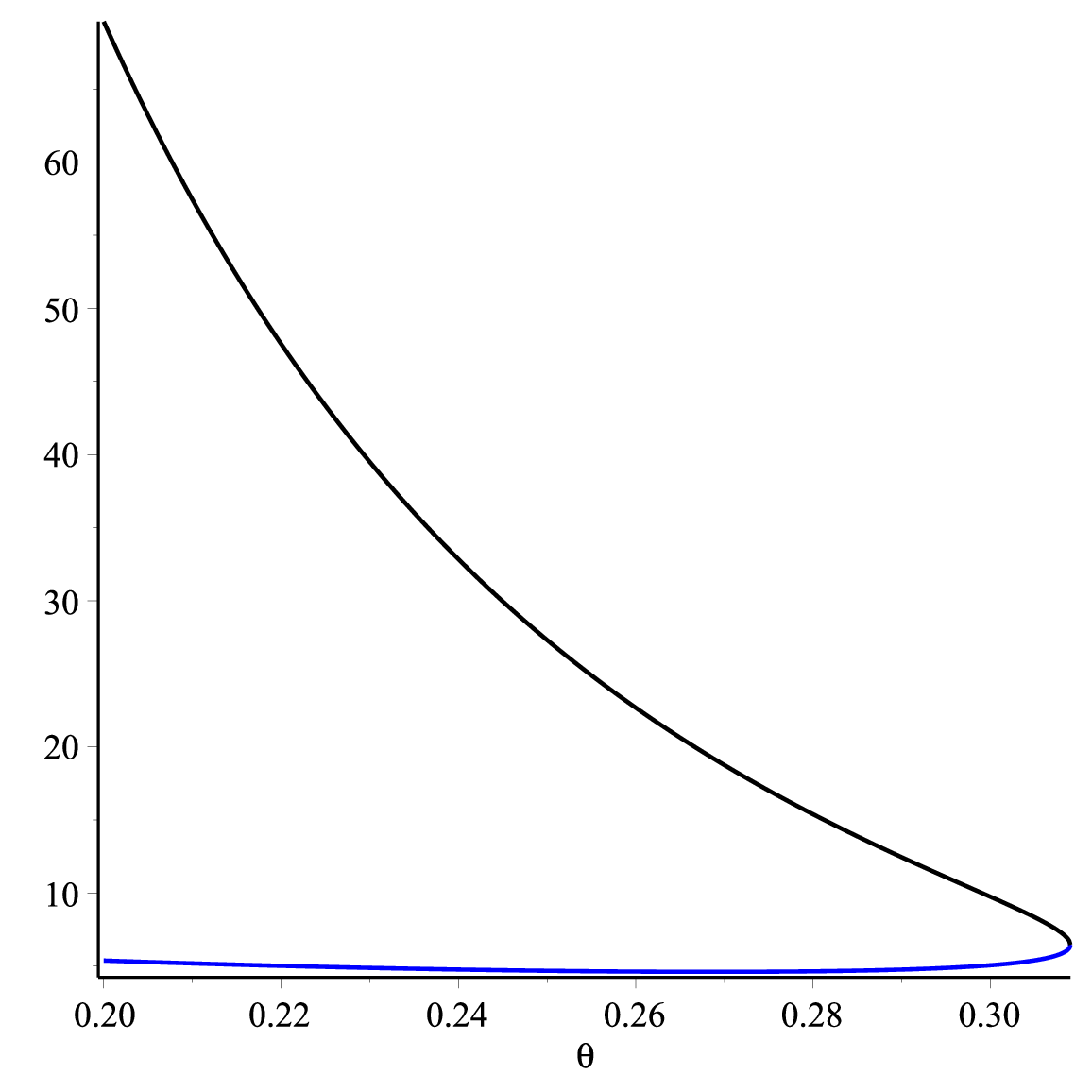} \ \ \ \ \
		\includegraphics[width=6.5cm]{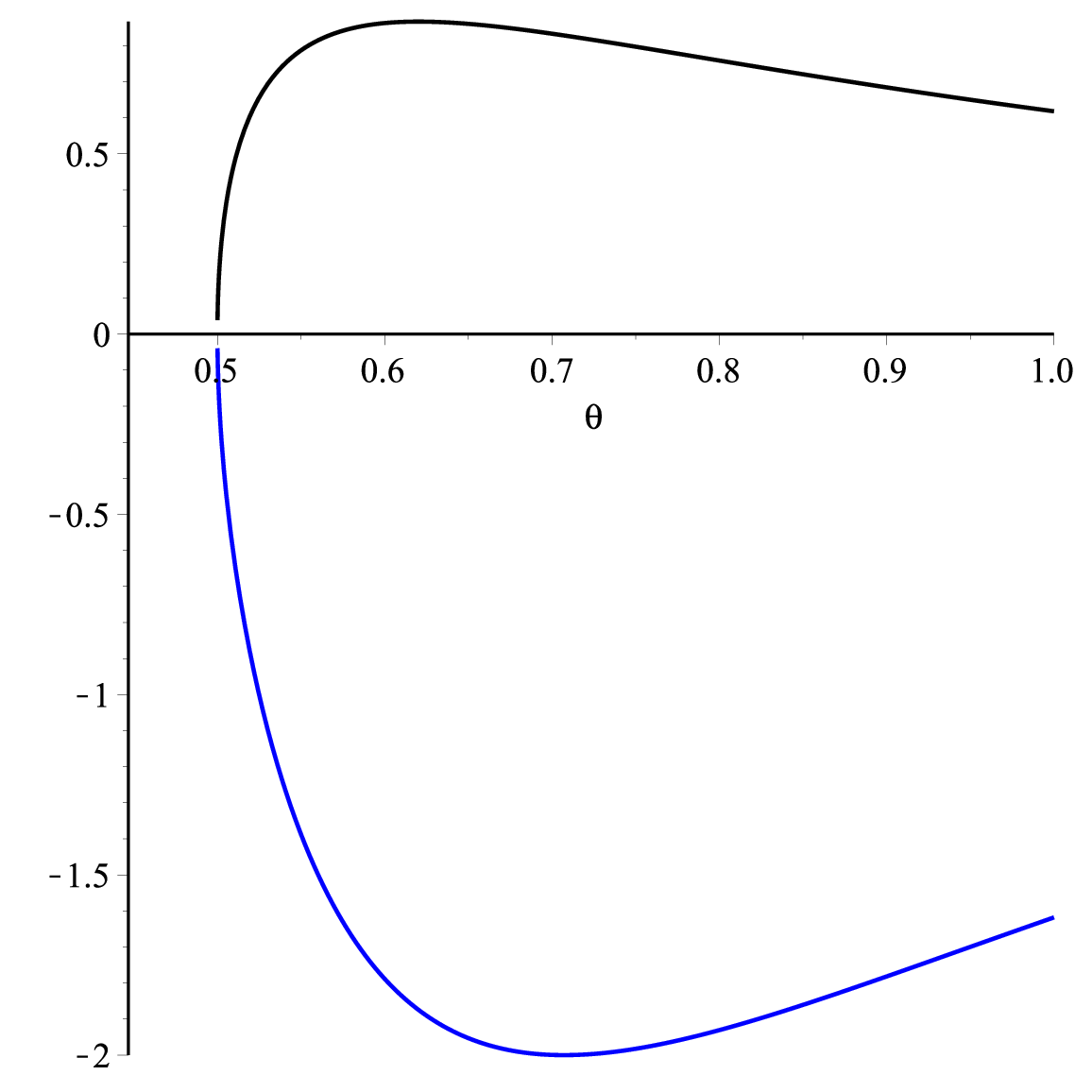}\ \ \
	\end{center}
	\caption{Left: The graphs of the functions $\theta\mapsto \xi_1(\theta)$ (blue) and $\theta\mapsto \xi_2(\theta)$ (black) for $\theta\in (0, (\sqrt{5}-1)/ 4]$. Right: The same graphs for $\theta\in [{1\over 2}, 1)$. }\label{x12}
\end{figure}
Moreover, the critical value of $\theta$ for existence of more than one solution is obtained from the discriminant of the cubic equation as $p\to \infty$, i.e.,  
$$\Delta_0(\theta)={1\over 27 \theta^2}\left(4(1-3\theta)^3-(2-9\theta+54\theta^3)^2\right)=0.$$
Hence, by Figure~\ref{y1-y3-p} it is clear that there exists a unique $\theta_0\approx 0.135$ such that $\Delta_0(\theta_0)=0$. 

\subsubsection{Case $x\ne 1$} 
In this case, there are up to four solutions when $p>0$ is fixed. These solutions are defined by the quantities $\xi_1(\theta,p)<\xi_2(\theta,p)$ given by
	\begin{equation}\label{12} 
 \xi_{1,2}(\theta,p):={q\over 2}{-3\theta q^2+2(\theta+1)q+2(\theta^2-1)\mp\theta\sqrt{q(q+2\theta-2)[(q-\theta-1)^2+(\theta+1)(3\theta-1)]}\over (q-\theta-1)[\theta q^2+(\theta^2-1)(q+\theta-1)]},
		\end{equation}
where $q=1-\theta^{2^p}$, see~\cite{JR}.
Moreover, if
\begin{equation}\label{C1}
	2<\xi_1(\theta, p)\leq \xi_2(\theta, p),
\end{equation} one can find all four positive solutions $x_i=x_i(\theta, p)$, $i=4,5,6,7$ explicitly.
In this case, as $p\to \infty$, we check the Condition~\eqref{C1}. From~\eqref{12} we get  
\begin{equation}\label{12i} 	
	\lim_{p\to \infty}\xi_{1,2}(\theta,p)=\xi_{1,2}(\theta):={1\over 2\theta^3}\cdot\left({1-2\theta\mp\sqrt{(2\theta-1)
(4\theta^2+2\theta-1)}}\right)
\end{equation}
and these numbers exist iff 
$$(2\theta-1)(4\theta^2+2\theta-1)\geq 0 \ \ \Leftrightarrow \ \ \theta\in (0, (\sqrt{5}-1)/4]\cup [1/2, 1).$$
Thus, Condition~\eqref{C1} is satisfied iff $\theta\in (0, (\sqrt{5}-1)/4]$ (see Figure~\ref{x12}). Now using~\eqref{12i}, for $\theta\in (0, (\sqrt{5}-1)/4))$, 
%from~\eqref{x4-7} 
we obtain $x_i(\theta)$, $i=4,5,6,7$. Since the last $x_i$'s exist, 
%taking the limit from both side of~\eqref{y4-7} 
we get 
$$y_i(\theta)=\sqrt{x_i(\theta)/\theta}, \ \ i=4,5,6,7.$$

\subsection{Case $\theta>1$} In this case, assuming $x=1$, 
%since 
%from~\eqref{De}, we get 
%$$\lim_{p\to\infty}\Delta(\theta,p)=-\infty$$
from ~\eqref{y3} we get a unique solution for large $p$, which has a limit as $p\to\infty$. 
If $\theta>1$ and $x\ne 1$ then the statement of Remark~\ref{<1} is satisfied for any $p>0$ and therefore there is no solution.
We summarize the results of this section in the following statement which essentially says that the number of translation-invariant SGMs remains unchanged in the limiting model as $p\to \infty$.
\begin{pro}\label{po}
	For the $p$-SOS model, as $p\to \infty$, there exist critical values $\theta_0\approx 0.135$ and $\theta_0'=(\sqrt{5}-1)/ 4\approx 0.309$ such that 
	\begin{itemize}
		\item[1.] If $\theta>\theta_0'$, then there is a unique translation-invariant SGM.
		\item[2.] If  $\theta=\theta_0'$, then there are three  translation-invariant SGMs.
		\item[3.] If  $\theta\in (\theta_0, \theta_0')$, then there are five  translation-invariant SGMs.
		\item[4.] If  $\theta=\theta_0$, then there are six   translation-invariant SGMs.
		\item[5.] If  $\theta\in (0, \theta_0)$, then there are seven translation-invariant SGMs.
	\end{itemize} 
\end{pro}
\section{Conditions for non-extremality of translation-invariant SGMs}

It is known that a translation-invariant SGM corresponding to a vector $v=(x,y)\in \mathbb{R}^2$ (which is a solution to~\eqref{si}) is a tree-indexed Markov chain with states $\{0,1,2\}$, see \cite[Definition 12.2]{Ge}, and for the transition matrix
\begin{equation}\label{m} {\mathbb P}=\left(\begin{array}{ccc}
{x^k\over x^k+\theta y^k}&{\theta y^k\over x^k+\theta y^k}&0\\[3mm]
{\theta x^k\over \theta x^k+y^k+\theta}&{y^k\over \theta x^k+y^k+\theta}&{\theta\over \theta x^k+y^k+\theta}\\[3mm]
0&{\theta y^k\over \theta y^k+1}&{1\over \theta y^k+1}
\end{array}
\right).
\end{equation}
Hence, for each given solution $(x_i,y_i)$, $i=1,\dots,7$ of~\eqref{si}, we need to calculate the eigenvalues of $\mathbb P$. The first eigenvalue is one since we deal with a stochastic matrix, the other two eigenvalues
 \begin{equation}\label{ev}\lambda_j(x_i, y_i, \theta, k), \qquad j=1,2,
\end{equation} 	
can be found via symbolic computer analysis, but they have bulky formulas. 
For example, in the case $x=1$, for each $y$ 
the matrix~\eqref{m} has three eigenvalues, 1 and
$$\lambda_1(1, y,\theta,k)={(1-2\theta^2)y^k\over \theta y^{2k}+(2\theta^2+1)y^k+2\theta}\qquad\text{and}\qquad
\lambda_2(1, y,\theta, k)=
{1\over \theta y^k+1}.
$$
However, we can still deduce the following relation.
\begin{lemma}\label{kam} If $\theta\in (0,1)$, then, for any solution $y$ of~\eqref{yk}, we have that
$$|\lambda_1(1,y,\theta,k)|\leq  \lambda_2(1,y,\theta,k).$$
\end{lemma}
\begin{proof} Since $\lambda_2>0$, we have to show that 
\begin{equation}\label{ten}
-\lambda_2(1,y,\theta,k)\leq\lambda_1(1,y,\theta,k)\leq  \lambda_2(1,y,\theta,k).
\end{equation}
It is easy to see that the inequality on the left is true for $\theta$ satisfying $1-2\theta^2\geq 0$. If $1-2\theta^2< 0$ then the inequality on the left is equivalent to 
$$\theta(1-\theta^2)y^{2k}+y^k+\theta\geq 0,$$ 
which is true for all $\theta<1$.
Next, the inequality on the right of~\eqref{ten} is equivalent to the inequality 
$$(\theta y^k+1)^2\geq 0,$$
which is universally true, concluding the proof.
\end{proof}

Now, a sufficient condition for non-extremality of a Gibbs measure $\mu$ corresponding
to ${\mathbb P}$ on a Cayley tree of order $k\geq 1$ is given by the Kesten--Stigum Condition
$k\lambda^2>1$, where $\lambda$ is the second-largest (in absolute value) eigenvalue of ${\mathbb P}$, see~\cite{Ke}.
Hence, denoting for $i=1,\dots,7$, 
%$$
%\lambda_{\max,i}(\theta,k)=\max\{|\lambda_1(x_i,y_i,\theta,k)|, |\lambda_2(x_i,y_i,\theta,k)|\},$$
$$\eta_i(\theta,k)=k\lambda^2_2(x_i,y_i,\theta,k)-1\qquad\text{ and }\qquad\mathbb K_i=\{(\theta, k)\in (0,1)\times \mathbb N\colon \eta_i(\theta,k)>0\},$$
using Lemma~\ref{kam}, we have the following criterion.
\begin{pro}\label{tne}
Let $\mu_i$ denote the translation-invariant SGM associated to the tuple  $(x_i,y_i,\theta,k)$. If $(\theta, k)\in \mathbb K_i$ then $\mu_i$ is non-extremal.
\end{pro}
In order to employ the proposition, 
for $k=2$ and $k=3$, we find representations for $\mathbb K_i$. 
%To do this we need explicite formula of solutions $(x_i, y_i)$. 
In case $k=2$ and $x_i=1$, we have for $i=1,2,3$ that $\lambda_2(1,y_i,\theta,2)=1/(\theta y_i^2+1)$ and thus 
$$\eta_i(\theta,2)={2\over (\theta y_i^2+1)^2}-1.$$
\begin{figure}[h]
	\begin{center}
\includegraphics[width=7.5cm]{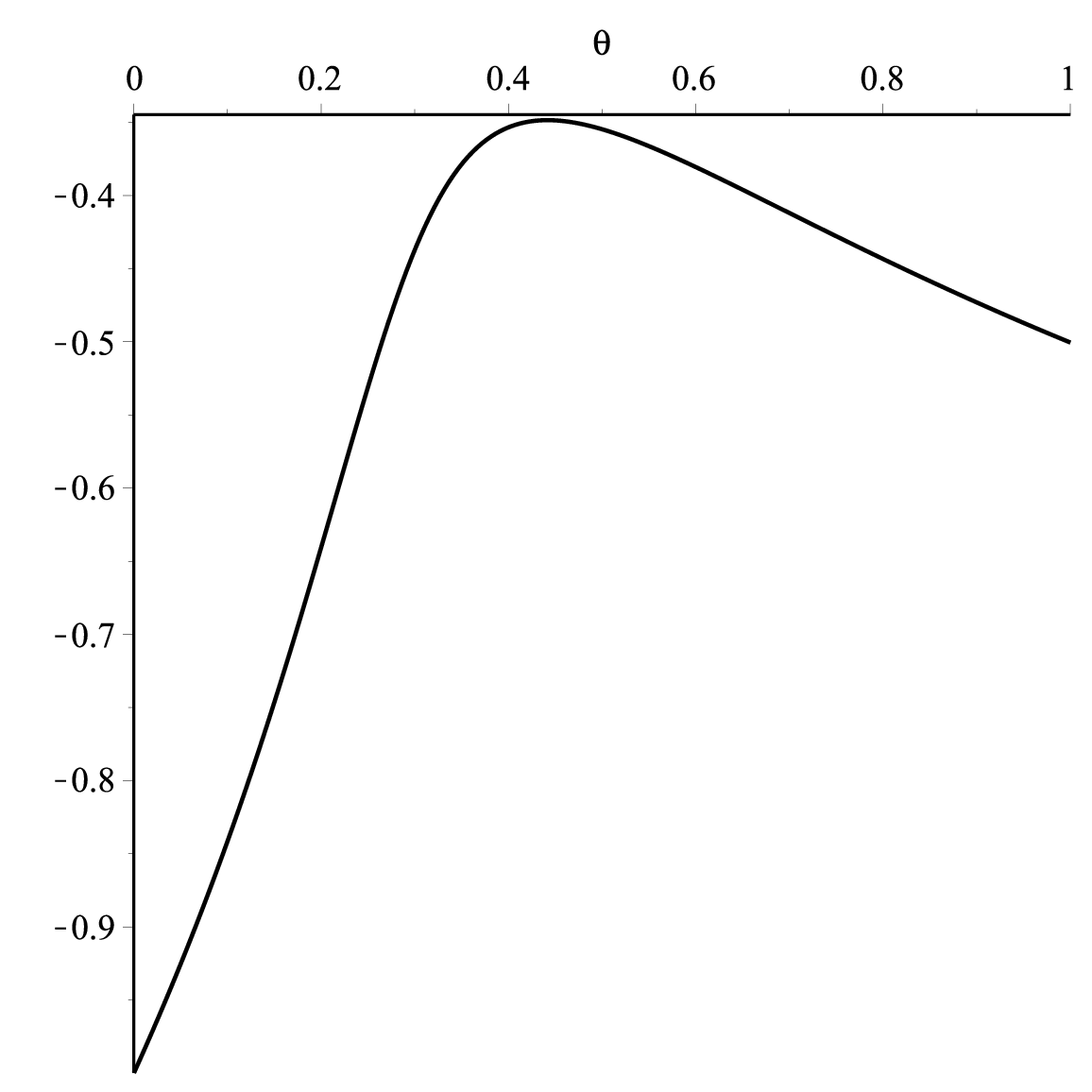} \ \ \ \ \
	\includegraphics[width=7.5cm]{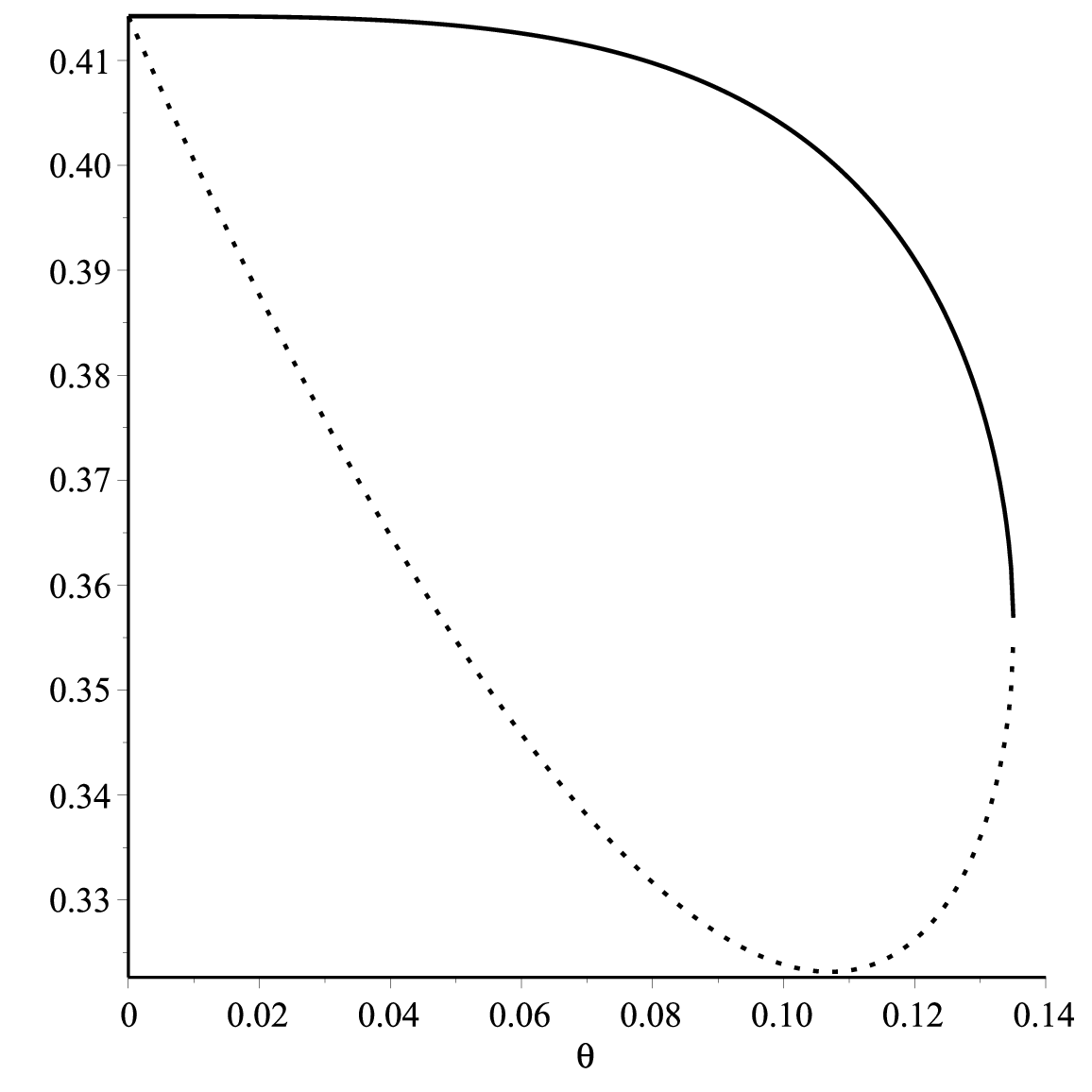} \ \ \
	\end{center}
	\caption{Left: The graph of the function $\eta_1(\theta,2)$, $\theta\in (0,1)$. Right: The graphs of the functions $\eta_2(\theta,2)$ (solid line) and  $\eta_3(\theta,2)$ (doted line) when $\theta\in(0, 0.14)$.}\label{mu123}
\end{figure}
Hence, from Figure~\ref{mu123} it follows that, for $\mu_1$,  the Kesten--Stigum condition is never satisfied, but for $\mu_2$ and $\mu_3$ the condition is always satisfied, i.e., $\mu_2$ and $\mu_3$ are not-extreme.

In case $k=3$ and $x_i=1$, we have for $i=1,2,3$, using $\lambda_2(1,y_i,\theta,3)=1/(\theta y_i^3+1)$, that
$$\eta_i(\theta,3)={3\over (\theta y_i^3+1)^2}-1.$$
But, as shown in Section~\ref{sec_13}, if $\theta<\theta_c\approx 0.206$, there exist two solutions $y_2,y_3<1$. For these solutions we have 
$$\eta_i(\theta,3)={3\over (\theta y_i^3+1)^2}-1> {3\over (\theta +1)^2}-1.$$
But, since $\theta<\theta_c$, we have that
$3/(\theta +1)^2-1>0$ and we can formulate the following summarizing result.
\begin{pro} For $\theta<1$,  $k=2$ and $k=3$ the translation-invariant SGMs corresponding to solutions of the form $(1,y)$ with $y<1$ are not-extreme.
\end{pro}
Let us note that for $k=3$ and $x_i=1$ the translation-invariant SGM corresponding to the solution $y_1>\sqrt[4]{2}$ does not satisfy the Kesten--Stigum condition if $\theta>(\sqrt{3}-1)/\sqrt[4]{8}\approx 0.435.$ Indeed, using $y_1>\sqrt[4]{2}$ we get 
$$\eta_i(\theta,3)={3\over (\theta y_1^3+1)^2}-1< {3\over (\theta \sqrt[4]{8}+1)^2}-1<0 \ \ \Leftrightarrow \ \ \theta>{\sqrt{3}-1\over \sqrt[4]{8}}.$$

\begin{rk} Let us finally discuss further extremality conditions for translation-invariant SGMs. Various approaches in the literature aim to establish sufficient conditions for extremality, which can be simplified to a finite-dimensional optimization problem based solely on the transition matrix. For instance, the percolation method proposed in~\cite{MSW} and \cite{Mos2}, the symmetric-entropy method by~\cite{FK}, or the bound provided in~\cite{Mar} for the Ising model in the presence of an external field. Different techniques are employed also in~\cite{BC} in order to demonstrate the sharpness of the Kesten-Stigum bound for an Ising channel with minimal asymmetry.

 However, since, in our case, the transition matrix corresponding to a translation-invariant SGM depends on the solutions $(x_i, y_i)$, which have a very complex form, it appears challenging to apply the aforementioned methods to verify extremality. Furthermore, the difficulty increases when we only have knowledge of the existence of a solution but lack its explicit form. Nonetheless, our results could serve as a basis for numerical investigations of extremality in the future.
\end{rk}

 \section*{Acknowledgements}

 B.~Jahnel is supported by the Leibniz Association within the Leibniz Junior Research Group on {\em Probabilistic Methods for Dynamic Communication Networks} as part of the Leibniz Competition (grant no.~J105/2020).
 U.~Rozikov thanks the Weierstrass Institute for Applied Analysis and Stochastics, Berlin, Germany for support of his visit.  His  work was partially supported through a grant from the IMU--CDC and the fundamental project (grant no.~F--FA--2021--425) of The Ministry of Innovative Development of the Republic of Uzbekistan.

\end{document}